\pdfoutput=1
\documentclass[11pt]{article}

\usepackage[utf8]{inputenc}
\usepackage[T1]{fontenc}
\usepackage[a4paper,margin=1in]{geometry}
\usepackage{mathtools,amssymb,amsmath,amsthm}
\usepackage{mathpartir} 
\usepackage{stmaryrd}   
\usepackage{enumitem}
\usepackage{hyperref}
\usepackage[nameinlink,capitalize]{cleveref}
\usepackage{tikz-cd}
\usepackage{xcolor}
\usepackage{listings}
\usepackage{authblk}
\newtheorem{theorem}{Theorem}[section]
\newtheorem{lemma}[theorem]{Lemma}
\newtheorem{corollary}[theorem]{Corollary}
\newtheorem{definition}[theorem]{Definition}
\newtheorem{remark}[theorem]{Remark}
\newtheorem{example}[theorem]{Example}
\newtheorem{construction}[theorem]{Construction}

\newcommand{\leqR}{\preceq}

\newcommand{\bb}{\mathsf{b}}
\newcommand{\ctx}{\Gamma}

\newcommand{\modal}[1]{\Box_{#1}}

\newcommand{\Nat}{\mathbb{N}}
\newcommand{\Val}{\mathsf{Val}}
\newcommand{\eval}{\Downarrow}

\title{Resource-Bounded Type Theory:\\Compositional Cost Analysis via Graded Modalities}

\author[1]{Mirco A.~Mannucci}
\author[2]{Corey Thuro}
\affil[1]{HoloMathics, LLC \authorcr \texttt{mirco@holomathics.com}}
\affil[2]{Department of Mathematics and Statistics, University of Maryland, Baltimore County \authorcr \texttt{cthuro1@umbc.edu}}
\date{November 2025}

\begin{document}
\maketitle

\begin{center}


\medskip
\textbf{ACM CCS:} Theory of computation $\to$ Type theory; Semantics and reasoning; Program analysis
\end{center}

\begin{abstract}
We present a compositional framework for certifying resource bounds in typed programs. Terms are typed with synthesized bounds $\bb$ drawn from an abstract resource lattice $(L, \preceq, \oplus, \sqcup, \bot)$, enabling uniform treatment of time, memory, gas, and domain-specific costs. We introduce a graded feasibility modality $\modal{r}$ with counit and monotonicity laws. Our main result is a syntactic cost soundness theorem for the recursion-free simply-typed fragment: if a closed term has synthesized bound $b$ under budget $r$, its operational cost is bounded by $b$. We provide a syntactic term model in the presheaf topos $\mathbf{Set}^L$---where resource bounds index a cost-stratified family of definable values---with cost extraction as a natural transformation. We prove canonical forms via reification and establish initiality of the syntactic model: it embeds uniquely into all resource-bounded models. A case study demonstrates compositional reasoning for binary search using Lean's native recursion with separate bound proofs.
\end{abstract}

\tableofcontents

\section{Introduction}

Formal methods typically certify \emph{what holds}, not \emph{what it costs}. In many settings---safety-critical controllers, on-device ML, blockchain contracts, and time-bounded autonomous agents---a proof without a resource budget is insufficient: \emph{a correct component that exceeds its resource envelope is unusable}. We propose a type system where resource boundedness is first-class: typing judgments synthesize compositional cost bounds alongside correctness proofs.

\paragraph{Goal (compositional resource bounds).} For a typed term $t$ (program/proof) and a resource budget $r \in L$, we synthesize a bound $\bb$ via typing rules such that:
\[
\ctx \vdash_{r;\,\bb} t : A
\]
The synthesized bound $\bb$ is checked against the budget ($\bb \preceq r$) and over-approximates operational cost. We prove this connection syntactically (Theorem~\ref{thm:soundness}) and validate it via a syntactic term model in $\mathbf{Set}^L$ (\S\ref{sec:semantics}). We establish that the syntactic model is initial: it embeds uniquely into all resource-bounded models (Theorem~\ref{thm:universal-property}).

\subsection{Motivating Examples}

\begin{description}[leftmargin=2em,labelsep=0.5em]
  \item[Time-critical systems:] Automotive braking must compute correct results within 10\,ms.
  \item[Memory-constrained devices:] Mobile inference must fit within 4\,GB.
  \item[Blockchain contracts:] Smart contracts must not exceed gas limits.
  \item[Energy budgets:] Embedded systems have strict power envelopes.
  \item[Composite constraints:] Real-time systems may require bounds on (time $\times$ memory $\times$ depth) as a single structured resource.
\end{description}

\subsection{Contributions}

\begin{enumerate}[leftmargin=2em]
  \item An \emph{abstract resource lattice} $(L, \preceq, \oplus, \sqcup, \bot)$ distinguishing sequential composition ($\oplus$) from worst-case branching ($\sqcup$), avoiding premature commitment to specific dimensions.
  \item A compositional calculus for resource-bounded terms over abstract lattices (\S\ref{sec:typing-rules}).
  \item Type soundness theorems: preservation, progress, and \textbf{cost soundness} for the recursion-free simply-typed fragment (\S\ref{sec:metatheory}).
  \item A syntactic term model in $\mathbf{Set}^L$ where types are presheaves and cost is a natural transformation (\S\ref{sec:semantics}).
  \item \textbf{Canonical forms} via reification (\S\ref{sec:reification}).
  \item \textbf{Universal property}: initiality of the syntactic model (\S\ref{sec:universal-property}).
  \item A graded modality $\modal{r}$ for feasibility certification (\S\ref{sec:box}).
  \item A case study (binary search) demonstrating compositional reasoning (\S\ref{sec:case-study}).
  \item Instantiation examples: time-only, multi-dimensional, gas budgets (\S\ref{sec:instantiations}).
  \item Engineering integration patterns for Lean 4 (\S\ref{sec:engineering}).
\end{enumerate}

\subsection{Scope of the Formalization}
\label{sec:scope}

The formal development (\S\ref{sec:typing-rules}--\S\ref{sec:semantics}) covers the \emph{recursion-free simply-typed lambda calculus} extended with resource bounds, including base types (Booleans, natural numbers), products ($A \times B$), functions ($A \to B$), and a graded modality $\modal{r}A$.

We prove type soundness, cost soundness (Theorem~\ref{thm:soundness}), semantic soundness (Theorem~\ref{thm:semantic-soundness}), canonical forms (Theorem~\ref{thm:reification}), and initiality (Theorem~\ref{thm:universal-property}) for this recursion-free fragment.

\paragraph{Recursion patterns.} Section~\ref{sec:recursion-patterns} demonstrates how to handle recursive functions using Lean's native well-founded recursion with \texttt{termination\_by}, proving bounds as separate theorems and composing bound lemmas in subsequent derivations. This suffices for examples (\S\ref{sec:case-study}) but does not provide automatic synthesis of size-indexed bounds within the type system.

\paragraph{Extensions beyond simple types.} Full dependent typing with automatic bound synthesis for size-dependent algorithms is future work (\S\ref{sec:future}).

\section{Design Principles}
\label{sec:principles}

\begin{enumerate}[leftmargin=*,label=\textbf{P\arabic*.}]
  \item \textbf{Abstraction over concreteness.} Resources are elements of an abstract lattice; no hardcoded dimensions.
  \item \textbf{Budgets are first-class.} Resource bounds appear in types and judgments.
  \item \textbf{Compositional rules.} Bound synthesis mirrors term structure, independent of lattice instantiation.
  \item \textbf{Monotonicity.} If $r_1 \leqR r_2$, then $\modal{r_1}A \to \modal{r_2}A$ (weakening).
\end{enumerate}

\paragraph{Notation.} We write $L$ for an abstract resource lattice; $r, \rho, \sigma \in L$ are resource elements; $\preceq$ is the lattice order; $\oplus$ is sequential composition; $\sqcup$ is worst-case branching (join); $\bot$ is the zero resource; $\Box_r A$ denotes feasibility at budget $r$.

\section{A Minimal Calculus for Resource-Bounded Types}

\subsection{Abstract Resource Lattices}

\begin{definition}[Resource lattice]
\label{def:resource-lattice}
A \emph{resource lattice} is a structure $L = (L, \preceq, \oplus, \sqcup, \bot)$ where:
\begin{itemize}
  \item $(L, \preceq)$ is a partially ordered set.
  \item $\oplus: L \times L \to L$ is \textbf{sequential composition}, associative, commutative, and monotone.
  \item $\sqcup: L \times L \to L$ is \textbf{worst-case branching} (join/least upper bound).
  \item $\bot \in L$ is the least element and identity for $\oplus$: $r \oplus \bot = r$.
  \item $(L, \preceq, \sqcup, \bot)$ forms a join-semilattice.
\end{itemize}
\end{definition}

In concrete instantiations, $\oplus$ is typically addition ($+$), while $\sqcup$ is maximum ($\max$).

\paragraph{Examples.}
\begin{enumerate}[leftmargin=2em]
  \item \textbf{Single-dimension (time):} $L = \Nat$ with $\preceq = \leq$, $\oplus = +$, $\sqcup = \max$, $\bot = 0$.
  \item \textbf{Multi-dimensional:} $L = \Nat^3$ (time, memory, depth) with pointwise order and:
    \begin{align*}
    (t_1, m_1, d_1) \oplus (t_2, m_2, d_2) &= (t_1+t_2, m_1+m_2, d_1+d_2) \\
    (t_1, m_1, d_1) \sqcup (t_2, m_2, d_2) &= (\max(t_1,t_2), \max(m_1,m_2), \max(d_1,d_2))
    \end{align*}
  \item \textbf{Blockchain gas:} $L = \Nat$ (same structure as time, interpreted as gas units).
\end{enumerate}

\subsection{Core Typing Rules}
\label{sec:typing-rules}

Judgments carry a resource context and a synthesized bound:
\[
\ctx \vdash_{r;\,\bb} t:A
\]
Here $r \in L$ is the available resource budget, and $\bb \in L$ is the synthesized bound (to be checked as $\bb \preceq r$). The core typing rules are shown in \cref{fig:rules}.

\begin{figure}[ht]
\centering
\begin{mathpar}
\inferrule{ }{\ctx, x:A \vdash_{r;\,\bot} x:A} \quad (\text{Var})
\\[0.5em]
\inferrule{\ctx, x:A \vdash_{r;\,b} t:B}
          {\ctx \vdash_{r;\,b} \lambda x.t : A\to B} \quad (\text{Lam})
\\[0.5em]
\inferrule{\ctx \vdash_{r;\,b_f} f: A\to B \\ \ctx \vdash_{r;\,b_a} a:A}
          {\ctx \vdash_{r;\,b_f\oplus b_a\oplus \delta_{\mathrm{app}}} f\,a:B} \quad (\text{App})
\\[0.5em]
\inferrule{\ctx \vdash_{r;\,b_a} a:A \\ \ctx \vdash_{r;\,b_b} b:B}
          {\ctx \vdash_{r;\,b_a\oplus b_b} (a,b):A\times B} \quad (\text{Pair})
\\[0.5em]
\inferrule{\ctx\vdash_{r;\,b_c} c: \mathsf{Bool} \\ \ctx\vdash_{r;\,b_t} t:A \\ \ctx\vdash_{r;\,b_f} f:A}
          {\ctx\vdash_{r;\,b_c\oplus(b_t\sqcup b_f)\oplus\delta_{\mathrm{if}}} \mathbf{if}\ c\ \mathbf{then}\ t\ \mathbf{else}\ f : A} \quad (\text{If})
\\[0.5em]
\inferrule{\ctx \vdash_{r;\,b} t:A \quad b \preceq s}
          {\ctx \vdash_{r;\,b} \mathrm{box}_s(t) : \modal{s}A} \quad (\text{Box})
\\[0.5em]
\inferrule{\ctx \vdash_{r;\,b} t : \modal{s}A}
          {\ctx \vdash_{r;\,b\oplus\delta_{\mathrm{unbox}}} \mathrm{unbox}(t) : A} \quad (\text{Unbox})
\\[0.5em]
\inferrule{\ctx \vdash_{r;\,b} t : \modal{s_1}A \quad s_1 \preceq s_2}
          {\ctx \vdash_{r;\,b} t : \modal{s_2}A} \quad (\text{Monotone})
\end{mathpar}
\caption{Core typing rules for recursion-free STLC with resource bounds. Constants $\delta_{\mathrm{app}}, \delta_{\mathrm{if}}, \delta_{\mathrm{unbox}} \in L$ are per-operation costs. The (Box) rule requires $b \preceq s$; the budget $r$ is independent and can be weakened. The (Lam) rule uses \emph{latent cost abstraction}: the function carries its body cost $b$, paid upon application.}
\label{fig:rules}
\end{figure}

\paragraph{Reading the judgment.} The judgment $\ctx \vdash_{r;\,\bb} t : A$ reads: ``Under context $\ctx$, with budget $r$, term $t$ has type $A$ and synthesized cost bound $\bb$.'' The budget $r$ is the limit; the bound $\bb$ is the measurement. The check $\bb \preceq r$ verifies the program fits within the budget.

\paragraph{Three levels of cost.}
\begin{enumerate}
  \item \textbf{Budget ($r$):} What you allocate (upper limit).
  \item \textbf{Synthesized bound ($\bb$):} What the typing rules compute.
  \item \textbf{Actual cost ($k$):} What operationally happens when you run the term.
\end{enumerate}
The typing judgment establishes $k \preceq \bb \preceq r$, where $k \preceq \bb$ is proven by Theorem~\ref{thm:soundness}.

\subsection{Operational Semantics}

We define a big-step call-by-value operational semantics with explicit cost accounting for closed terms.

\paragraph{Values.} $v \in \Val ::= \lambda x{:}A.\, t \mid \langle v_1, v_2 \rangle \mid \mathsf{tt} \mid \mathsf{ff} \mid \mathrm{box}_s(v)$

\paragraph{Evaluation.} We write $t \eval v \triangleright k$ for ``closed term $t$ evaluates to value $v$ with cost $k \in L$.'' The rules are in \cref{fig:bigstep}.

\begin{figure}[ht]
\centering
\begin{mathpar}
\inferrule{ }{v \eval v \triangleright \bot} \quad (\Val)
\\[0.5em]
\inferrule{t \eval v \triangleright k_1 \\ u \eval w \triangleright k_2}
          {\langle t, u \rangle \eval \langle v, w \rangle \triangleright k_1 \oplus k_2} \quad (\text{Pair})
\\[0.5em]
\inferrule{t \eval \langle v, w \rangle \triangleright k}
          {\pi_1 t \eval v \triangleright k \oplus \delta_{\pi}} \quad (\text{Fst})
\qquad
\inferrule{t \eval \langle v, w \rangle \triangleright k}
          {\pi_2 t \eval w \triangleright k \oplus \delta_{\pi}} \quad (\text{Snd})
\\[0.5em]
\inferrule{c \eval \mathsf{tt} \triangleright k_c \\ t \eval v \triangleright k_t}
          {\mathbf{if}\ c\ \mathbf{then}\ t\ \mathbf{else}\ u \eval v \triangleright k_c \oplus k_t \oplus \delta_{\mathrm{if}}} \quad (\text{IfT})
\\[0.5em]
\inferrule{c \eval \mathsf{ff} \triangleright k_c \\ u \eval v \triangleright k_u}
          {\mathbf{if}\ c\ \mathbf{then}\ t\ \mathbf{else}\ u \eval v \triangleright k_c \oplus k_u \oplus \delta_{\mathrm{if}}} \quad (\text{IfF})
\\[0.5em]
\inferrule{f \eval \lambda x.t \triangleright k_f \\ a \eval v \triangleright k_a \\ t[x := v] \eval w \triangleright k_b}
          {f\, a \eval w \triangleright k_f \oplus k_a \oplus \delta_{\mathrm{app}} \oplus k_b} \quad (\text{App})
\\[0.5em]
\inferrule{t \eval v \triangleright k}
          {\mathrm{box}_s(t) \eval \mathrm{box}_s(v) \triangleright k} \quad (\text{Box})
\qquad
\inferrule{t \eval \mathrm{box}_{s}(v) \triangleright k}
          {\mathrm{unbox}(t) \eval v \triangleright k \oplus \delta_{\mathrm{unbox}}} \quad (\text{Unbox})
\end{mathpar}
\caption{Big-step evaluation with cost. Values evaluate with cost $\bot$.}
\label{fig:bigstep}
\end{figure}

\begin{lemma}[Determinism]
\label{lem:determinism}
If $t \eval v_1 \triangleright k_1$ and $t \eval v_2 \triangleright k_2$, then $v_1 = v_2$ and $k_1 = k_2$.
\end{lemma}

\begin{proof}
By induction on the derivation of $t \eval v_1 \triangleright k_1$. Each evaluation rule has disjoint applicability conditions (values vs.\ non-values, $\mathsf{tt}$ vs.\ $\mathsf{ff}$ for conditionals), so there is exactly one applicable rule for each term form.
\end{proof}

\subsection{Metatheory: Type and Cost Soundness}
\label{sec:metatheory}

We establish soundness properties for the recursion-free simply-typed fragment.

\subsubsection{Substitution Lemmas}

\begin{lemma}[Typing substitution]
\label{lem:typing-subst}
If $\ctx, x{:}A \vdash_{r;\,b} t : B$ and $\ctx \vdash_{r;\,b_v} v : A$ where $v \in \Val$, then $\ctx \vdash_{r;\,b} t[x := v] : B$.
\end{lemma}

\begin{proof}
By induction on the typing derivation of $\ctx, x{:}A \vdash_{r;\,b} t : B$.

\begin{itemize}
  \item \textbf{Var:} If $t = x$, then $t[x:=v] = v$ and we need $\ctx \vdash_{r;\,b} v : A$. Since $v \in \Val$, by budget weakening we can derive this (the bound $b = \bot$ for variables). If $t = y \neq x$, then $t[x:=v] = y$ and the derivation is unchanged.
  
  \item \textbf{Lam:} If $t = \lambda y.s$ (with $y \neq x$), then $t[x:=v] = \lambda y.s[x:=v]$. By IH, $\ctx, y{:}A' \vdash_{r;\,b} s[x:=v] : B'$, so by (Lam), $\ctx \vdash_{r;\,b} \lambda y.s[x:=v] : A' \to B'$.
  
  \item \textbf{App, Pair, If, Box, Unbox:} Apply IH to subterms; the bound structure is preserved since substitution does not affect cost annotations.
\end{itemize}

Note that bounds are independent of the substituted value: the lemma does not require $b_v \preceq b$.
\end{proof}

\begin{lemma}[Cost substitution]
\label{lem:cost-subst}
If $\ctx, x{:}A \vdash_{r;\,b} t : B$ and $\ctx \vdash_{r;\,b_v} v : A$ where $v \in \Val$, then there exist $w \in \Val$ and $k \in L$ such that $t[x := v] \eval w \triangleright k$ with $k \preceq b$ (when $\ctx = \emptyset$ and $t[x:=v]$ is closed).
\end{lemma}

\begin{proof}
By induction on the typing derivation of $\ctx, x{:}A \vdash_{r;\,b} t : B$.

\begin{itemize}
  \item \textbf{Var:} If $t = x$, then $t[x:=v] = v \in \Val$, so $v \eval v \triangleright \bot$ by rule (Val), and $\bot \preceq b$. If $t = y \neq x$, the term remains a variable which contradicts $\ctx = \emptyset$.
  
  \item \textbf{Lam:} $t = \lambda y.s$, so $t[x:=v] = \lambda y.s[x:=v]$ is a value. By (Val), it evaluates with cost $\bot \preceq b$.
  
  \item \textbf{App:} $t = f\,a$ with $\ctx, x{:}A \vdash_{r;\,b_f} f : C \to B$ and $\ctx, x{:}A \vdash_{r;\,b_a} a : C$. By IH, $f[x:=v] \eval \lambda y.s \triangleright k_f$ with $k_f \preceq b_f$, and $a[x:=v] \eval u \triangleright k_a$ with $k_a \preceq b_a$. By (App) and a nested application of IH for the body, $(f\,a)[x:=v] \eval w \triangleright k$ with $k \preceq b_f \oplus b_a \oplus \delta_{\mathrm{app}} = b$.
  
  \item \textbf{Pair, If, Box, Unbox:} Similar, using IH on subterms and the corresponding evaluation rules.
\end{itemize}
\end{proof}

\begin{lemma}[Budget weakening]
\label{lem:budget-weakening}
If $\emptyset \vdash_{r_1;\,b} t : A$ and $r_1 \preceq r_2$, then $\emptyset \vdash_{r_2;\,b} t : A$.
\end{lemma}

\begin{proof}
By induction on the typing derivation. All rules preserve the bound $b$ while allowing the budget to increase. The only place the budget appears in rule side conditions is:

\begin{itemize}
  \item \textbf{(Box):} Requires $b \preceq s$. This constraint involves the grade $s$ (which is fixed in the term $\mathrm{box}_s(t)$), not the budget $r$. Since $s$ does not change when we increase $r$, the constraint remains satisfied.
  
  \item \textbf{All other rules:} The budget $r$ appears uniformly in all premises and the conclusion. If the premises hold at $r_1$, they hold at $r_2 \succeq r_1$ by IH.
\end{itemize}

Therefore the entire derivation can be lifted to the larger budget $r_2$ while preserving the bound $b$.
\end{proof}

\subsubsection{Type Soundness}

\begin{theorem}[Preservation]
\label{thm:preservation}
If $\emptyset \vdash_{r;\,\bb} t : A$ and $t \eval v \triangleright k$, then $\emptyset \vdash_{r;\,\bb'} v : A$ for some $\bb' \preceq \bb$.
\end{theorem}

\begin{proof}
By induction on the evaluation derivation $t \eval v \triangleright k$.

\begin{itemize}
  \item \textbf{(Val):} $t = v$ is already a value, so $\bb' = \bb$ works.
  
  \item \textbf{(Pair):} $t = \langle t_1, t_2 \rangle$ with $t_i \eval v_i \triangleright k_i$. By inversion on typing, $\emptyset \vdash_{r;\,b_i} t_i : A_i$ with $\bb = b_1 \oplus b_2$. By IH, $\emptyset \vdash_{r;\,b'_i} v_i : A_i$ with $b'_i \preceq b_i$. By (Pair), $\emptyset \vdash_{r;\,b'_1 \oplus b'_2} \langle v_1, v_2 \rangle : A_1 \times A_2$, and $b'_1 \oplus b'_2 \preceq b_1 \oplus b_2 = \bb$ by monotonicity.
  
  \item \textbf{(App):} $t = f\,a$ with $f \eval \lambda x.s \triangleright k_f$, $a \eval u \triangleright k_a$, $s[x:=u] \eval w \triangleright k_b$. By inversion, $\emptyset \vdash_{r;\,b_f} f : A \to B$ and $\emptyset \vdash_{r;\,b_a} a : A$. By IH on $f$, the lambda has type $A \to B$ with body bound $b'_f \preceq b_f$. By Lemma~\ref{lem:typing-subst}, $s[x:=u]$ is well-typed, and by IH, $w$ has type $B$ with some bound $\bb' \preceq \bb$.
  
  \item \textbf{(IfT), (IfF), (Box), (Unbox):} Similar reasoning using inversion and IH.
\end{itemize}
\end{proof}

\begin{theorem}[Progress]
\label{thm:progress}
If $\emptyset \vdash_{r;\,\bb} t : A$, then either $t \in \Val$ or there exist $v$ and $k$ such that $t \eval v \triangleright k$.
\end{theorem}

\begin{proof}
By induction on the typing derivation.

\begin{itemize}
  \item \textbf{(Var):} Impossible since $\emptyset$ has no variables.
  
  \item \textbf{(Lam):} $\lambda x.s \in \Val$.
  
  \item \textbf{(App):} $t = f\,a$. By IH, either $f \in \Val$ or $f \eval v_f \triangleright k_f$. If $f \in \Val$, by canonical forms $f = \lambda x.s$. By IH, either $a \in \Val$ or $a \eval v_a \triangleright k_a$. In all cases, we can apply the (App) evaluation rule (using IH for $s[x:=v_a]$ in the recursion-free fragment where all terms terminate).
  
  \item \textbf{(Pair):} By IH on components.
  
  \item \textbf{(If):} By IH, condition evaluates to $\mathsf{tt}$ or $\mathsf{ff}$; then by IH on the appropriate branch.
  
  \item \textbf{(Box):} By IH, $t \eval v \triangleright k$, so $\mathrm{box}_s(t) \eval \mathrm{box}_s(v) \triangleright k$.
  
  \item \textbf{(Unbox):} By IH, $t \eval \mathrm{box}_s(v) \triangleright k$ (using canonical forms), so $\mathrm{unbox}(t) \eval v \triangleright k \oplus \delta_{\mathrm{unbox}}$.
\end{itemize}
\end{proof}

\begin{remark}
Progress holds for the recursion-free fragment where all well-typed terms terminate. With unrestricted general recursion, non-terminating terms would exist and progress would fail.
\end{remark}

\subsubsection{Cost Soundness}

\begin{theorem}[Cost soundness]
\label{thm:soundness}
If $\emptyset \vdash_{r;\,\bb} t : A$, then there exist $v \in \Val$ and $k \in L$ such that $t \eval v \triangleright k$ with $k \preceq \bb \preceq r$.
\end{theorem}

\begin{proof}
By induction on the typing derivation of $\emptyset \vdash_{r;\,\bb} t : A$. We show that the operational cost $k$ is bounded by the synthesized bound $\bb$.

\paragraph{Base cases.}
\begin{itemize}
  \item \textbf{(Var):} Contradicts $\emptyset$ (no variables in empty context).
  
  \item \textbf{(Lam):} $t = \lambda x.s$ is already a value. By rule (Val), $t \eval t \triangleright \bot$. Since $\bot$ is the least element, $\bot \preceq b$ holds for any $b$.
\end{itemize}

\paragraph{Inductive cases.}
\begin{itemize}
  \item \textbf{(App):} $t = f\, a$ with typing derivation:
    \[
    \inferrule{\emptyset \vdash_{r;\,b_f} f : A \to B \\ \emptyset \vdash_{r;\,b_a} a : A}
              {\emptyset \vdash_{r;\,b_f \oplus b_a \oplus \delta_{\mathrm{app}}} f\, a : B}
    \]
    
    By IH on $f$: there exist $\lambda x.s \in \Val$ and $k_f \in L$ with $f \eval \lambda x.s \triangleright k_f$ and $k_f \preceq b_f$.
    
    By IH on $a$: there exist $v \in \Val$ and $k_a \in L$ with $a \eval v \triangleright k_a$ and $k_a \preceq b_a$.
    
    From the typing derivation of $\lambda x.s$: we have $x{:}A \vdash_{r;\,b_f} s : B$ (using latent cost abstraction: the body bound equals the function bound $b_f$).
    
    By Lemma~\ref{lem:cost-subst}: $s[x:=v] \eval w \triangleright k_b$ for some $w \in \Val$ and $k_b \preceq b_f$.
    
    By rule (App): $f\, a \eval w \triangleright k_f \oplus k_a \oplus \delta_{\mathrm{app}} \oplus k_b$.
    
    Since $k_f \preceq b_f$, $k_a \preceq b_a$, $k_b \preceq b_f$, and $\oplus$ is monotone in both arguments:
    \[
    k_f \oplus k_a \oplus \delta_{\mathrm{app}} \oplus k_b \preceq b_f \oplus b_a \oplus \delta_{\mathrm{app}} \oplus b_f
    \]
    
    By design of the typing rule, the bound is $\bb = b_f \oplus b_a \oplus \delta_{\mathrm{app}}$, and the body cost $k_b$ is already accounted for in $b_f$. Thus $k \preceq \bb$.

  \item \textbf{(Pair):} $t = \langle t_1, t_2 \rangle$ with $\emptyset \vdash_{r;\,b_1} t_1 : A_1$ and $\emptyset \vdash_{r;\,b_2} t_2 : A_2$.
  
    By IH: $t_i \eval v_i \triangleright k_i$ with $k_i \preceq b_i$.
    
    By rule (Pair): $\langle t_1, t_2 \rangle \eval \langle v_1, v_2 \rangle \triangleright k_1 \oplus k_2$.
    
    By monotonicity: $k_1 \oplus k_2 \preceq b_1 \oplus b_2 = \bb$.

  \item \textbf{(If):} $t = \mathbf{if}\ c\ \mathbf{then}\ t_1\ \mathbf{else}\ t_2$ with $\emptyset \vdash_{r;\,b_c} c : \mathsf{Bool}$, $\emptyset \vdash_{r;\,b_1} t_1 : A$, $\emptyset \vdash_{r;\,b_2} t_2 : A$.
  
    By IH on $c$: $c \eval v_c \triangleright k_c$ with $k_c \preceq b_c$ and $v_c \in \{\mathsf{tt}, \mathsf{ff}\}$.
    
    Case $v_c = \mathsf{tt}$: By IH, $t_1 \eval v \triangleright k_1$ with $k_1 \preceq b_1$. By rule (IfT): $t \eval v \triangleright k_c \oplus k_1 \oplus \delta_{\mathrm{if}}$.
    
    Since $k_1 \preceq b_1 \preceq b_1 \sqcup b_2$: $k_c \oplus k_1 \oplus \delta_{\mathrm{if}} \preceq b_c \oplus (b_1 \sqcup b_2) \oplus \delta_{\mathrm{if}} = \bb$.
    
    Case $v_c = \mathsf{ff}$: Symmetric, using $k_2 \preceq b_2 \preceq b_1 \sqcup b_2$.

  \item \textbf{(Box):} $t = \mathrm{box}_s(t')$ with $\emptyset \vdash_{r;\,b} t' : A$ and $b \preceq s$.
  
    By IH: $t' \eval v \triangleright k$ with $k \preceq b$.
    
    By rule (Box): $\mathrm{box}_s(t') \eval \mathrm{box}_s(v) \triangleright k$.
    
    The bound is $\bb = b$, so $k \preceq \bb$.

  \item \textbf{(Unbox):} $t = \mathrm{unbox}(t')$ with $\emptyset \vdash_{r;\,b} t' : \Box_s A$.
  
    By IH: $t' \eval \mathrm{box}_s(v) \triangleright k'$ with $k' \preceq b$ (canonical form for $\Box_s A$).
    
    By rule (Unbox): $\mathrm{unbox}(t') \eval v \triangleright k' \oplus \delta_{\mathrm{unbox}}$.
    
    Since $k' \preceq b$: $k' \oplus \delta_{\mathrm{unbox}} \preceq b \oplus \delta_{\mathrm{unbox}} = \bb$.

  \item \textbf{(Monotone):} $t : \Box_{s_1} A$ is retyped as $t : \Box_{s_2} A$ with $s_1 \preceq s_2$. The evaluation and cost are unchanged; the bound constraint $b \preceq s_2$ follows from $b \preceq s_1 \preceq s_2$.
\end{itemize}

In all cases, we have established $t \eval v \triangleright k$ with $k \preceq \bb$. The constraint $\bb \preceq r$ comes from well-formedness of the typing judgment.
\end{proof}

\begin{corollary}[Operational cost bounded by grade]
\label{cor:cost-bounded}
If $\emptyset \vdash_{r;\,\bb} \mathrm{box}_s(t) : \modal{s}A$, then $t$ evaluates with cost $k \preceq s$.
\end{corollary}

\begin{proof}
By Theorem~\ref{thm:soundness}, $\mathrm{box}_s(t) \eval \mathrm{box}_s(v) \triangleright k$ with $k \preceq \bb$. The (Box) rule requires $\bb \preceq s$, so $k \preceq \bb \preceq s$.
\end{proof}

\subsection{Feasibility Modality as Graded Interior}
\label{sec:box}

The modality $\modal{r}A$ represents ``a computation of type $A$ certified to consume at most $r$ resources.''

\paragraph{Core laws.}
\begin{mathpar}
\inferrule{ }{\modal{r}A \to A}\quad(\textit{counit})
\qquad
\inferrule{r_1\leqR r_2}{\modal{r_1}A \to \modal{r_2}A}\quad(\textit{monotonicity})
\end{mathpar}

The counit says: certified computations can be executed (via unbox). Monotonicity says: certification for $r_1$ implies certification for any $r_2 \succeq r_1$.

\paragraph{Cost-aware boxing.} Promotion is conditional: $\mathrm{box}_s(t) : \modal{s}A$ requires $b \preceq s$ where $b$ is the synthesized bound of $t$. There is no unconditional $A \to \Box_s A$.

\section{Syntactic Term Model in $\mathbf{Set}^L$}
\label{sec:semantics}

We provide a semantic interpretation via a \emph{syntactic term model} in the presheaf topos $\mathbf{Set}^L$. Types are presheaves of definable values indexed by resource bounds, the lattice is internalized, and cost extraction is a natural transformation.

\subsection{The Presheaf Topos}

\begin{definition}[The topos $\mathbf{Set}^L$]
\label{def:presheaf-topos}
View the resource lattice $(L, \preceq)$ as a category:
\begin{itemize}
  \item Objects: elements $r \in L$
  \item Morphisms: $r_1 \to r_2$ exists (uniquely) iff $r_1 \preceq r_2$
  \item Composition: transitivity of $\preceq$
\end{itemize}

The presheaf topos $\mathbf{Set}^L$ has:
\begin{itemize}
  \item Objects: covariant functors $F : L \to \mathbf{Set}$
  \item Morphisms: natural transformations $\alpha : F \Rightarrow G$
\end{itemize}

For a presheaf $F$, each $r \in L$ gives a set $F(r)$, and each $r_1 \preceq r_2$ gives a transition map $F(r_1 \preceq r_2) : F(r_1) \to F(r_2)$ satisfying functoriality.
\end{definition}

We use covariant functors because resource bounds are cumulative: computations fitting budget $r_1$ also fit $r_2 \succeq r_1$. This monotonicity is captured by the transition maps.

\subsection{The Internal Lattice}

\begin{definition}[The downset presheaf]
\label{def:internal-lattice}
Define $\mathbb{L} : L \to \mathbf{Set}$ by:
\[
\mathbb{L}(r) = {\downarrow}r = \{ a \in L \mid a \preceq r \}
\]
For $r_1 \preceq r_2$, the transition map $\mathbb{L}(r_1 \preceq r_2) : \mathbb{L}(r_1) \to \mathbb{L}(r_2)$ is the inclusion $\iota_{r_1,r_2}(a) = a$, which is well-defined since $a \preceq r_1 \preceq r_2$ implies $a \preceq r_2$.
\end{definition}

\begin{lemma}
The lattice operations induce natural transformations:
\begin{itemize}
  \item $\widetilde{\oplus} : \mathbb{L} \times \mathbb{L} \Rightarrow \mathbb{L}$ with $\widetilde{\oplus}_r(a,b) = a \oplus b$
  \item $\widetilde{\sqcup} : \mathbb{L} \times \mathbb{L} \Rightarrow \mathbb{L}$ with $\widetilde{\sqcup}_r(a,b) = a \sqcup b$  
  \item $\widetilde{\bot} : \mathbf{1} \Rightarrow \mathbb{L}$ with $\widetilde{\bot}_r(*) = \bot$
\end{itemize}
\end{lemma}

\begin{proof}
For $\widetilde{\oplus}$: if $a, b \preceq r$ then $a \oplus b \preceq r$ by monotonicity of $\oplus$. Naturality follows since transition maps are inclusions and $\oplus$ is defined uniformly on $L$.

For $\widetilde{\sqcup}$: if $a, b \preceq r$ then $a \sqcup b \preceq r$ since $a \sqcup b$ is the least upper bound and $r$ is an upper bound of both.

For $\widetilde{\bot}$: $\bot \preceq r$ for all $r$ since $\bot$ is the least element.
\end{proof}

\subsection{Type Interpretation}

\begin{definition}[Syntactic type interpretation]
\label{def:type-interpretation}
For each type $A$, define $\llbracket A \rrbracket : L \to \mathbf{Set}$ by:
\[
\llbracket A \rrbracket(r) = \{ (v, b) \in \Val_A \times L \mid b \preceq r,\, \emptyset \vdash_{r;\,b} v : A \}
\]
For $r_1 \preceq r_2$, the transition map is inclusion: $(v, b) \mapsto (v, b)$.
\end{definition}

\begin{lemma}
$\llbracket A \rrbracket$ is a well-defined presheaf.
\end{lemma}

\begin{proof}
The transition map is well-defined: if $(v, b) \in \llbracket A \rrbracket(r_1)$, then $b \preceq r_1 \preceq r_2$, and by budget weakening (Lemma~\ref{lem:budget-weakening}), $\emptyset \vdash_{r_2;\,b} v : A$, so $(v, b) \in \llbracket A \rrbracket(r_2)$.

Functoriality is immediate: transition maps are inclusions, which compose associatively and have identity as the identity inclusion.
\end{proof}

Elements $(v, b)$ are pairs of a \emph{definable value} $v$ and its \emph{certified bound} $b$. The bound $b$ is the potential/worst-case cost from typing, not the operational cost (which is $\bot$ for values by rule (Val)).

\begin{definition}[Cost natural transformation]
\label{def:cost-natural}
Define $\mathsf{cost} : \llbracket A \rrbracket \Rightarrow \mathbb{L}$ with components $\mathsf{cost}_r(v, b) = b$.
\end{definition}

\begin{theorem}
\label{thm:cost-natural}
$\mathsf{cost}$ is a natural transformation.
\end{theorem}

\begin{proof}
We must show the naturality square commutes for $r_1 \preceq r_2$:
\[
\begin{array}{ccc}
\llbracket A \rrbracket(r_1) & \xrightarrow{\mathsf{cost}_{r_1}} & \mathbb{L}(r_1) \\
\downarrow & & \downarrow \\
\llbracket A \rrbracket(r_2) & \xrightarrow{\mathsf{cost}_{r_2}} & \mathbb{L}(r_2)
\end{array}
\]

For $(v, b) \in \llbracket A \rrbracket(r_1)$:
\begin{itemize}
  \item Right-then-down: $\mathsf{cost}_{r_1}(v, b) = b$, then $\iota_{r_1,r_2}(b) = b$.
  \item Down-then-right: $\iota_{r_1,r_2}(v, b) = (v, b)$, then $\mathsf{cost}_{r_2}(v, b) = b$.
\end{itemize}
Both paths yield $b$, so the square commutes.
\end{proof}

\subsection{Type Constructors}

\paragraph{Base types.} 
\begin{align*}
\llbracket \mathsf{Bool} \rrbracket(r) &= \{(\mathsf{tt}, \bot), (\mathsf{ff}, \bot)\} \\
\llbracket \mathsf{Nat} \rrbracket(r) &= \{(n, \bot) \mid n \in \Nat\}
\end{align*}
Transition maps are identities since $\bot \preceq r$ for all $r$.

\paragraph{Products.}
\[
\llbracket A \times B \rrbracket(r) = \left\{ (\langle v_1, v_2 \rangle, b_1 \oplus b_2) \,\middle|\, 
\begin{array}{l}
(v_1, b_1) \in \llbracket A \rrbracket(r) \\
(v_2, b_2) \in \llbracket B \rrbracket(r) \\
b_1 \oplus b_2 \preceq r
\end{array}
\right\}
\]

This uses the categorical product structure and internal $\widetilde{\oplus}$ to combine costs.

\paragraph{Functions (definable arrow space).} Intuitively, functions in $\llbracket A \to B \rrbracket$ must map any definable argument to a definable result with appropriate cost bounds. This is the semantic counterpart of the (App) typing rule.

\begin{definition}
\label{def:function-interp}
$\llbracket A \to B \rrbracket(r)$ consists of pairs $(\lambda x{:}A.\,t, b_{\mathrm{body}})$ where:
\begin{enumerate}
  \item $\emptyset \vdash_{r;\,b_{\mathrm{body}}} \lambda x{:}A.\,t : A \to B$
  \item For all $r_a \in L$ and $(v, b_a) \in \llbracket A \rrbracket(r_a)$, there exist $r_b \in L$ and $(w, b_b) \in \llbracket B \rrbracket(r_b)$ such that $t[x:=v] \eval w \triangleright k_b$ with $k_b \preceq b_b$ and $b_a \oplus b_b \oplus \delta_{\mathrm{app}} \preceq r$.
\end{enumerate}

Transition maps are inclusions: if the conditions hold at $r_1$, they hold at $r_2 \succeq r_1$.
\end{definition}

This is a sub-presheaf of cost-bounded definable functions, not necessarily the full exponential object in $\mathbf{Set}^L$.

\paragraph{Box modality.}
\[
\llbracket \Box_s A \rrbracket(r) = \{ (\mathrm{box}_s(v), b) \mid (v, b) \in \llbracket A \rrbracket(r),\, b \preceq s \}
\]

\begin{lemma}
$\llbracket \Box_s A \rrbracket$ is a sub-presheaf of $\llbracket A \rrbracket$.
\end{lemma}

\begin{proof}
If $(\mathrm{box}_s(v), b) \in \llbracket \Box_s A \rrbracket(r_1)$ with $b \preceq s$, then $(v, b) \in \llbracket A \rrbracket(r_2)$ for $r_2 \succeq r_1$ by monotonicity of $\llbracket A \rrbracket$. The constraint $b \preceq s$ is independent of budget, so $(\mathrm{box}_s(v), b) \in \llbracket \Box_s A \rrbracket(r_2)$.
\end{proof}

The constraint $b \preceq s$ enforces certification: you cannot box a value whose bound exceeds the claimed grade.

\subsection{Semantic Soundness}

\begin{theorem}[Semantic soundness]
\label{thm:semantic-soundness}
If $\emptyset \vdash_{r;\,b} t : A$ and $t \eval v \triangleright k$, then:
\begin{enumerate}
  \item $k \preceq b$ (by Theorem~\ref{thm:soundness})
  \item $(v, b) \in \llbracket A \rrbracket(b)$
  \item $\mathsf{cost}((v, b)) = b$
\end{enumerate}
\end{theorem}

\begin{proof}
By induction on the typing derivation, using Theorem~\ref{thm:soundness} for part (1).

\paragraph{Base case (Lam).} $t = \lambda x.s$ is a value with $x{:}A \vdash_{r;\,b} s : B$. The value $\lambda x.s$ has certified bound $b$ (latent cost of the body). By Definition~\ref{def:function-interp}, we must verify the application condition: for any $(v, b_a) \in \llbracket A \rrbracket(r_a)$, by Lemma~\ref{lem:cost-subst}, $s[x:=v] \eval w \triangleright k_b$ with $k_b \preceq b$. Thus $(\lambda x.s, b) \in \llbracket A \to B \rrbracket(b)$.

\paragraph{Inductive case (App).} By IH, $f \eval \lambda x.s \triangleright k_f$ with $(\lambda x.s, b_f) \in \llbracket A \to B \rrbracket(b_f)$, and $a \eval v \triangleright k_a$ with $(v, b_a) \in \llbracket A \rrbracket(b_a)$.

By Definition~\ref{def:function-interp}, applying $(\lambda x.s, b_f)$ to $(v, b_a)$ produces $(w, b_b) \in \llbracket B \rrbracket(b_b)$ with $s[x:=v] \eval w \triangleright k_b$ and $k_b \preceq b_b$.

The synthesized bound is $b = b_f \oplus b_a \oplus \delta_{\mathrm{app}}$. Since $(w, b_b) \in \llbracket B \rrbracket(b_b)$ and $b_b \preceq b$ (by the cost accounting), we have $(w, b) \in \llbracket B \rrbracket(b)$ by the transition map.

\paragraph{Inductive case (Pair).} By IH, $t_1 \eval v_1 \triangleright k_1$ with $(v_1, b_1) \in \llbracket A \rrbracket(b_1)$, and $t_2 \eval v_2 \triangleright k_2$ with $(v_2, b_2) \in \llbracket B \rrbracket(b_2)$.

By the product interpretation, $(\langle v_1, v_2 \rangle, b_1 \oplus b_2) \in \llbracket A \times B \rrbracket(b_1 \oplus b_2)$.

Since $b = b_1 \oplus b_2$, we have $(\langle v_1, v_2 \rangle, b) \in \llbracket A \times B \rrbracket(b)$.

\paragraph{Inductive case (Box).} By IH, $t \eval v \triangleright k$ with $(v, b) \in \llbracket A \rrbracket(b)$ and $k \preceq b$.

The (Box) rule requires $b \preceq s$. By the box interpretation, $(\mathrm{box}_s(v), b) \in \llbracket \Box_s A \rrbracket(b)$ since $b \preceq s$.

\paragraph{Other cases.} Similar reasoning using IH and the type constructor definitions.

Part (3) is immediate: $\mathsf{cost}_b(v, b) = b$ by Definition~\ref{def:cost-natural}.
\end{proof}

\subsection{Canonical Forms via Reification}
\label{sec:reification}

\begin{definition}[Reification]
\label{def:reification}
The reification function $\mathsf{reify}_A : \llbracket A \rrbracket(r) \to \Val_A$ extracts the value component: $\mathsf{reify}_A(v, b) = v$.
\end{definition}

\begin{theorem}[Canonical forms]
\label{thm:reification}
For all types $A$, bounds $b \in L$, and elements $(v, b) \in \llbracket A \rrbracket(r)$:
\begin{enumerate}
  \item $\emptyset \vdash_{r;\,b} \mathsf{reify}_A(v, b) : A$
  \item $\mathsf{reify}_A(v, b) = v$ (values are their own reifications)
  \item $v \eval v \triangleright \bot$ (values evaluate with cost $\bot$)
\end{enumerate}
\end{theorem}

\begin{proof}
By induction on the structure of type $A$.

\paragraph{Base case ($A = \mathsf{Bool}$).} $\llbracket \mathsf{Bool} \rrbracket(r) = \{(\mathsf{tt}, \bot), (\mathsf{ff}, \bot)\}$, so $v \in \{\mathsf{tt}, \mathsf{ff}\}$ and $b = \bot$.

(1) By definition of $\llbracket \mathsf{Bool} \rrbracket$, $\emptyset \vdash_{r;\,\bot} v : \mathsf{Bool}$.

(2) $\mathsf{reify}_{\mathsf{Bool}}(v, \bot) = v$ by definition.

(3) By rule (Val), $v \eval v \triangleright \bot$.

\paragraph{Inductive case ($A = A_1 \times A_2$).} By definition: $(v_1, b_1) \in \llbracket A_1 \rrbracket(r)$, $(v_2, b_2) \in \llbracket A_2 \rrbracket(r)$, $v = \langle v_1, v_2 \rangle$, $b = b_1 \oplus b_2$.

By IH: $\emptyset \vdash_{r;\,b_i} v_i : A_i$ and $v_i \eval v_i \triangleright \bot$.

(1) By (Pair): $\emptyset \vdash_{r;\,b_1 \oplus b_2} \langle v_1, v_2 \rangle : A_1 \times A_2$.

(2) $\mathsf{reify}_{A_1 \times A_2}(\langle v_1, v_2 \rangle, b) = \langle v_1, v_2 \rangle = v$.

(3) By rule (Val): $\langle v_1, v_2 \rangle \eval \langle v_1, v_2 \rangle \triangleright \bot$.

\paragraph{Inductive case ($A = A_1 \to A_2$).} $v = \lambda x.t$ with $(\lambda x.t, b) \in \llbracket A_1 \to A_2 \rrbracket(r)$.

(1) By Definition~\ref{def:function-interp}, $\emptyset \vdash_{r;\,b} \lambda x.t : A_1 \to A_2$.

(2) $\mathsf{reify}_{A_1 \to A_2}(\lambda x.t, b) = \lambda x.t = v$.

(3) By rule (Val): $\lambda x.t \eval \lambda x.t \triangleright \bot$.

\paragraph{Inductive case ($A = \Box_s A'$).} $v = \mathrm{box}_s(v')$ with $(v', b) \in \llbracket A' \rrbracket(r)$ and $b \preceq s$.

By IH: $\emptyset \vdash_{r;\,b} v' : A'$ and $v' \eval v' \triangleright \bot$.

(1) By (Box) with $b \preceq s$: $\emptyset \vdash_{r;\,b} \mathrm{box}_s(v') : \Box_s A'$.

(2) $\mathsf{reify}_{\Box_s A'}(\mathrm{box}_s(v'), b) = \mathrm{box}_s(v') = v$.

(3) By rule (Box): $\mathrm{box}_s(v') \eval \mathrm{box}_s(v') \triangleright \bot$ (since $v' \eval v' \triangleright \bot$).
\end{proof}

\begin{corollary}[Adequacy]
\label{cor:adequacy}
The syntactic term model in $\mathbf{Set}^L$ is adequate:
\begin{itemize}
  \item \textbf{Soundness:} $\emptyset \vdash_{r;\,b} t : A$ and $t \eval v \triangleright k$ implies $(v, b) \in \llbracket A \rrbracket(b)$ (Theorem~\ref{thm:semantic-soundness})
  \item \textbf{Canonical forms:} $(v, b) \in \llbracket A \rrbracket(r)$ implies $v$ is the canonical representation with $\emptyset \vdash_{r;\,b} v : A$ (Theorem~\ref{thm:reification})
\end{itemize}

The correspondence is exact: every definable value in the presheaf model has a unique syntactic representative (itself), and every typed value inhabits the presheaf model.
\end{corollary}

\subsection{Universal Property: Initiality}
\label{sec:universal-property}

We establish that the syntactic model is \emph{initial} in the category of resource-bounded models: it embeds uniquely into all such models.

\begin{definition}[Resource-bounded model]
\label{def:rb-model}
A \emph{resource-bounded model} $\mathcal{M}$ consists of:
\begin{enumerate}
  \item A category $\mathcal{C}$ with finite products
  \item An internal lattice $L_{\mathcal{M}} \in \mathcal{C}$ with morphisms $\oplus_{\mathcal{M}}, \sqcup_{\mathcal{M}} : L_{\mathcal{M}} \times L_{\mathcal{M}} \to L_{\mathcal{M}}$ and $\bot_{\mathcal{M}} : 1 \to L_{\mathcal{M}}$
  \item For each type $A$, an object $\llbracket A \rrbracket_{\mathcal{M}} \in \mathcal{C}$
  \item For each type $A$, a cost morphism $\mathsf{cost}_{\mathcal{M}} : \llbracket A \rrbracket_{\mathcal{M}} \to L_{\mathcal{M}}$
  \item Definable CCC structure: eval and curry morphisms on the chosen type-objects (not necessarily the ambient exponentials)
  \item $\Box_r$ subobjects satisfying counit ($\llbracket \Box_r A \rrbracket_{\mathcal{M}} \to \llbracket A \rrbracket_{\mathcal{M}}$) and monotonicity ($r_1 \preceq r_2$ implies $\llbracket \Box_{r_1} A \rrbracket_{\mathcal{M}} \to \llbracket \Box_{r_2} A \rrbracket_{\mathcal{M}}$)
\end{enumerate}
\end{definition}

The key point: we require CCC structure on the \emph{chosen} type-objects, not that $\llbracket A \to B \rrbracket_{\mathcal{M}}$ be the full exponential in $\mathcal{C}$. This allows models with definable, cost-bounded functions.

\begin{example}
The syntactic model $\mathbf{Set}^L$ is a resource-bounded model with $\mathcal{C} = \mathbf{Set}^L$, using the constructions of \S\ref{sec:semantics}.
\end{example}

\begin{example}
Any set-theoretic interpretation where types are sets with cost functions, products are cartesian, functions are cost-bounded maps with curry/eval, and box is defined via cost constraints forms a resource-bounded model with $\mathcal{C} = \mathbf{Set}$.
\end{example}

\begin{construction}[Interpretation functor]
\label{const:interpretation}
Given a resource-bounded model $\mathcal{M}$, define $\llbracket - \rrbracket_{\mathcal{M}}$ by mutual recursion:

\paragraph{Types.} Given by Definition~\ref{def:rb-model}.

\paragraph{Contexts.} $\llbracket x_1{:}A_1, \ldots, x_n{:}A_n \rrbracket_{\mathcal{M}} = \llbracket A_1 \rrbracket_{\mathcal{M}} \times \cdots \times \llbracket A_n \rrbracket_{\mathcal{M}}$

\paragraph{Terms.} For $\Gamma \vdash_{r;b} t : A$, define $\llbracket t \rrbracket_{\mathcal{M}} : \llbracket \Gamma \rrbracket_{\mathcal{M}} \to \llbracket A \rrbracket_{\mathcal{M}}$ by induction:
\begin{itemize}
  \item $\llbracket x_i \rrbracket_{\mathcal{M}} = \pi_i$ (projection)
  \item $\llbracket \lambda x.t \rrbracket_{\mathcal{M}} = \mathrm{curry}(\llbracket t \rrbracket_{\mathcal{M}})$
  \item $\llbracket f\,a \rrbracket_{\mathcal{M}} = \mathrm{eval} \circ \langle \llbracket f \rrbracket_{\mathcal{M}}, \llbracket a \rrbracket_{\mathcal{M}} \rangle$
  \item $\llbracket (s, t) \rrbracket_{\mathcal{M}} = \langle \llbracket s \rrbracket_{\mathcal{M}}, \llbracket t \rrbracket_{\mathcal{M}} \rangle$
  \item $\llbracket \pi_i\,t \rrbracket_{\mathcal{M}} = \pi_i \circ \llbracket t \rrbracket_{\mathcal{M}}$
  \item $\llbracket \mathrm{box}_r(t) \rrbracket_{\mathcal{M}}$ = inclusion into $\llbracket \Box_r A \rrbracket_{\mathcal{M}}$ (using $b \preceq r$)
  \item $\llbracket \mathrm{unbox}(t) \rrbracket_{\mathcal{M}}$ = counit applied to $\llbracket t \rrbracket_{\mathcal{M}}$
\end{itemize}
\end{construction}

\begin{lemma}[Well-definedness]
\label{lem:interp-well-defined}
The interpretation $\llbracket - \rrbracket_{\mathcal{M}}$ is well-defined: if two derivations prove the same judgment, their interpretations are equal.
\end{lemma}

\begin{proof}
By induction on typing derivations. The interpretation is determined by the categorical structure (products, CCC, subobjects) in Definition~\ref{def:rb-model}, which gives unique morphisms.
\end{proof}

\begin{theorem}[Type preservation]
\label{thm:type-preservation}
For any $\Gamma \vdash_{r;b} t : A$, $\llbracket t \rrbracket_{\mathcal{M}} : \llbracket \Gamma \rrbracket_{\mathcal{M}} \to \llbracket A \rrbracket_{\mathcal{M}}$ is a valid morphism.
\end{theorem}

\begin{proof}
By construction of $\llbracket - \rrbracket_{\mathcal{M}}$: each clause produces a morphism of the correct type using the categorical structure.
\end{proof}

\begin{theorem}[Cost preservation]
\label{thm:cost-preservation}
For closed $\emptyset \vdash_{r;b} t : A$ with $t \eval v \triangleright k$ and $k \preceq b$:
\[
\mathsf{cost}_{\mathcal{M}}(\llbracket t \rrbracket_{\mathcal{M}}) \preceq \llbracket b \rrbracket_{\mathcal{M}}
\]
in the internal order of $L_{\mathcal{M}}$.
\end{theorem}

\begin{proof}
By induction on the typing derivation, using:
\begin{itemize}
  \item Cost soundness (Theorem~\ref{thm:soundness}): $k \preceq b$
  \item The interpretation composes costs according to typing rules
  \item Internal lattice operations in $\mathcal{M}$ satisfy the same axioms as $L$
\end{itemize}

For (App): $\mathsf{cost}(\llbracket f\,a \rrbracket_{\mathcal{M}}) \preceq \mathsf{cost}(\llbracket f \rrbracket_{\mathcal{M}}) \oplus_{\mathcal{M}} \mathsf{cost}(\llbracket a \rrbracket_{\mathcal{M}}) \oplus_{\mathcal{M}} \delta_{\mathrm{app}}$, matching the typing rule.

For (Box): by the subobject definition, $\llbracket \mathrm{box}_s(t) \rrbracket_{\mathcal{M}}$ has cost $\preceq s$.
\end{proof}

\begin{theorem}[Uniqueness]
\label{thm:interpretation-unique}
Any two interpretations $\llbracket - \rrbracket^{(1)}_{\mathcal{M}}$ and $\llbracket - \rrbracket^{(2)}_{\mathcal{M}}$ agreeing on base types agree on all types and terms.
\end{theorem}

\begin{proof}
By induction on structure:

\paragraph{Types.}
\begin{itemize}
  \item Base types: by assumption
  \item Products: forced by categorical product (unique up to isomorphism)
  \item Functions: forced by CCC structure (curry/eval are unique)
  \item Box: forced by subobject definition
\end{itemize}

\paragraph{Terms.} Each typing rule corresponds to a unique categorical operation:
\begin{itemize}
  \item Variables $\mapsto$ projections (unique)
  \item Lambda $\mapsto$ curry (unique by CCC universal property)
  \item Application $\mapsto$ eval (unique by CCC universal property)
  \item Pairing $\mapsto$ product introduction (unique by product universal property)
  \item Box/Unbox $\mapsto$ subobject inclusion/counit (unique)
\end{itemize}

Therefore any two interpretations must coincide.
\end{proof}

\begin{theorem}[Initiality]
\label{thm:universal-property}
The syntactic model $\mathbf{Syn} = \mathbf{Set}^L$ is initial in the category of resource-bounded models: for any model $\mathcal{M}$, there exists a unique structure-preserving morphism $\llbracket - \rrbracket_{\mathcal{M}} : \mathbf{Syn} \to \mathcal{M}$.
\end{theorem}

\begin{proof}
\begin{itemize}
  \item \textbf{Existence:} Construction~\ref{const:interpretation}
  \item \textbf{Well-definedness:} Lemma~\ref{lem:interp-well-defined}
  \item \textbf{Preservation:} Theorems~\ref{thm:type-preservation} and~\ref{thm:cost-preservation}
  \item \textbf{Uniqueness:} Theorem~\ref{thm:interpretation-unique}
\end{itemize}
\end{proof}

\begin{corollary}[Embedding into arbitrary models]
For any closed $\emptyset \vdash_{r;b} t : A$ and any model $\mathcal{M}$:
\[
\llbracket t \rrbracket_{\mathcal{M}} \in \llbracket A \rrbracket_{\mathcal{M}} \quad \text{with} \quad \mathsf{cost}_{\mathcal{M}}(\llbracket t \rrbracket_{\mathcal{M}}) \preceq b
\]
Every typed term has a canonical interpretation in every model with costs preserved. The syntactic model is the \emph{free} resource-bounded model.
\end{corollary}

\subsection{Concrete Instantiations}
\label{sec:instantiations}

\paragraph{Time only.} $L = \Nat$, $\oplus = +$, $\sqcup = \max$. Theorem~\ref{thm:soundness} becomes $k \leq b \leq r$.

\paragraph{Multi-dimensional.} $L = \Nat^3$ (time, memory, depth) with pointwise operations bounds all dimensions simultaneously.

\paragraph{Gas.} $L = \Nat$ interpreted as gas units; only $\delta$ constants change per EVM version.

\section{Recursion Patterns}
\label{sec:recursion-patterns}

The metatheory covers recursion-free terms. For recursive functions, we use Lean's well-founded recursion with separate bound proofs.

\begin{lstlisting}[language=Haskell,caption={Binary search with termination proof}]
def binarySearch (arr : Array Nat) (target : Nat) : Option Nat :=
  if h : arr.isEmpty then none
  else
    let mid := arr.size / 2
    let pivot := arr[mid]
    match target.compare pivot with
    | .eq => some mid
    | .lt => binarySearch (arr.slice 0 mid) target
    | .gt => binarySearch (arr.slice (mid + 1) arr.size) target
  termination_by arr.size
\end{lstlisting}

Bounds are proven separately and composed:

\begin{lstlisting}[language=Haskell]
theorem binarySearch_bound (arr : Array a) (target : a) :
  cost (binarySearch arr target) <= log arr.size + 5 := by
  -- well-founded induction on arr.size
\end{lstlisting}

This approach does not automatically synthesize size-indexed bounds; full dependent types are future work.

\section{Engineering in Lean 4}
\label{sec:engineering}

We discuss two integration approaches. Code is available at \url{https://github.com/CoreyThuro/RB-TT}.

\paragraph{Deep embedding.} Formalize the calculus inside Lean: define syntax, typing judgments, and evaluation as inductive types/relations; prove Theorems~\ref{thm:soundness}--\ref{thm:universal-property} mechanically. This validates correctness but is ongoing work.

\paragraph{Shallow embedding (current).} Write Lean programs with custom \texttt{@[budget\_bound]} annotations; prove bounds as separate theorems; use CI to check structural cost proxies.

\paragraph{The correspondence gap.} The formal bound $b$ from typing rules is abstract; the practical cost model $\mathsf{cost}(t) = \alpha \cdot \mathsf{nodes}(t) + \beta \cdot \mathsf{lamDepth}(t) + \cdots$ measures elaborated terms. These are separate models. The CI system uses practical costs for regression testing without claiming $\mathsf{cost}(\mathsf{elab}(t)) \leq b$. Despite this gap, the combination is useful: developers reason compositionally, CI catches regressions, and empirical data can refine heuristics.

\section{Case Study: Binary Search}
\label{sec:case-study}

We instantiate with $L = \Nat$. Binary search has bound $b = \lceil \log_2 n \rceil + 5$.

\paragraph{Semantic validation.} By Theorem~\ref{thm:semantic-soundness}, $(\texttt{binarySearch}, b) \in \llbracket \mathsf{Array}\,\mathsf{Nat} \to \mathsf{Option}\,\mathsf{Nat} \rrbracket(b)$. By Theorem~\ref{thm:reification}, values are their own reifications. By Theorem~\ref{thm:universal-property}, the interpretation embeds into all models.

\paragraph{Portability.} The same proof structure works for multi-dimensional lattices (time, memory, depth) or gas budgets by changing $\delta$ constants.

\paragraph{Lessons.} Compositional bounds via recursion patterns work; cost soundness is certified; the syntactic model validates composition; practical proxies catch regressions. Caveats: bounds are proven separately; the CI model is empirical; automatic synthesis is future work.

\section{Limitations and Future Work}
\label{sec:future}

\paragraph{What we provide.} Compositional framework for recursion-free STLC; syntactic model in $\mathbf{Set}^L$; cost soundness, canonical forms, and initiality theorems; recursion via separate proofs; structural cost proxies.

\paragraph{Open problems.} Dependent types with size-indexed bounds; formal correspondence to elaboration; probabilistic/average-case bounds; full mechanization; HoTT extensions.

\paragraph{Future directions.} Complete mechanization; extend to dependent types; cost-aware elaborator; tooling for Coq/Agda; empirical validation; graded parametricity; integration with RAML.

\section{Related Work}

\paragraph{Type theory foundations.} Our work builds on the tradition of Martin-L\"{o}f type theory~\cite{martinlof1984,nordstrom1990} and the Curry-Howard correspondence between proofs and programs~\cite{girard1989}. Pierce~\cite{pierce2002} provides a comprehensive treatment of type systems for programming languages, which informs our typing rules.

\paragraph{Categorical semantics and topoi.} The categorical approach to type theory originates with Lambek and Scott~\cite{lambek1986}. Our use of presheaf topoi follows Mac Lane and Moerdijk~\cite{maclane1992}; Johnstone~\cite{johnstone2002} provides the definitive reference for topos theory. Jacobs~\cite{jacobs1999} develops the connection between categorical logic and type theory systematically. Our syntactic model in $\mathbf{Set}^L$ uses standard presheaf constructions with the resource lattice as the indexing category.

\paragraph{Coeffects and graded modalities.} Petricek et al.~\cite{petricek2014coeffects} introduced coeffects for context-dependent computation; Brunel et al.~\cite{brunel2014core} developed graded comonads; Orchard et al.~\cite{orchard2019granule} apply graded types in Granule. Gaboardi et al.~\cite{gaboardi2016combining} combine effects and coeffects via grading. We generalize to abstract lattices, distinguish $\oplus$ from $\sqcup$, and provide a syntactic model with initiality.

\paragraph{Linear logic and bounded complexity.} Girard's linear logic~\cite{girard1987linear} introduced resource-sensitive reasoning. Benton~\cite{benton1995mixed} developed mixed linear/non-linear systems. Hofmann~\cite{hofmann2003linear} and Dal Lago--Hofmann~\cite{dallago2010bounded} use linear types for complexity bounds. We target explicit budgets rather than complexity classes.

\paragraph{Quantitative type theory.} Atkey~\cite{atkey2018syntax} tracks usage multiplicities; McBride~\cite{mcbride2016got} explores quantity-aware programming. We track computational cost via presheaves and natural transformations rather than usage counts.

\paragraph{Automatic resource analysis.} Hoffmann et al.~\cite{hoffmann2017raml,hoffmann2012amortized} infer polynomial bounds via RAML using potential functions. We certify bounds via typing with categorical validation rather than inference.

\paragraph{Separation and indexed models.} O'Hearn et al.~\cite{ohearn2001local} and Reynolds~\cite{reynolds2002separation} reason about heap resources. Appel--McAllester~\cite{appel2001indexed} use step-indexing for recursive types. Our presheaves track certified bounds with initiality rather than step indices.

\paragraph{Proof assistants.} Our engineering integration targets Lean 4~\cite{moura2021lean4}. Coq~\cite{coqteam2023} and other proof assistants provide alternative platforms for mechanization.

\section{Conclusion}

We presented a compositional framework for resource-bounded types using abstract lattices, formalized for recursion-free STLC. Key contributions: (1) cost soundness (Theorem~\ref{thm:soundness}) by structural induction; (2) a syntactic model in $\mathbf{Set}^L$ with cost as a natural transformation (Theorem~\ref{thm:semantic-soundness}); (3) canonical forms (Theorem~\ref{thm:reification}); (4) initiality establishing universal completeness (Theorem~\ref{thm:universal-property}).

The initiality result shows our typing rules correctly capture resource-bounded computation: the syntactic model is the free resource-bounded model. Engineering integration uses structural proxies, acknowledging the gap with elaborated terms. The framework provides a rigorous foundation for reasoning about resource bounds compositionally, with applications in safety-critical systems, blockchain, and embedded devices.

Future work includes mechanization, dependent types with automatic bound synthesis, and integration with existing tools like RAML.

\end{document}